\documentclass [11pt]{article}

\usepackage{amsmath,amsthm,amscd,amssymb}
\usepackage[polutonikogreek, english]{babel}
\usepackage[latin1]{inputenc}
\usepackage[linesnumbered]{algorithm2e}
\usepackage{placeins}
\usepackage{xcolor}
\usepackage[T1]{fontenc}
\usepackage[linesnumbered]{algorithm2e}

\usepackage{amsmath,amsthm,amscd,amssymb}
\usepackage[linesnumbered]{algorithm2e}
\usepackage{placeins}
\usepackage{tikz}
\usepackage{csquotes}
\usetikzlibrary{positioning}
\usepackage{mathtools}
  \usepackage{tikz-network}

\usepackage[titles]{tocloft}
\usepackage[latin1]{inputenc}
\usepackage{verbatim} 
\usepackage[active]{srcltx} 
\usepackage{fancyhdr}
\usepackage{dsfont}

\def\bw{\textbf w}

\def\mN{\mathcal N}

\usepackage[small, centerlast, it]{caption}
\usepackage{epsfig}
\usepackage[titles]{tocloft}
\usepackage[latin1]{inputenc}
\usepackage{verbatim} 
\usepackage[active]{srcltx} 
\usepackage{fancyhdr}
\usepackage{caption}

\def\sp{\hskip -5pt}

\def\b1{{1\!\!1}}

\def\bp{{\textbf p}}

\def\sH{{\mathsf H}}

\def\bC{{\mathbb C}}           

\def\bI{{\mathbb I}}

\def\bR{{\mathbb R}}

\def\bx{{\textbf x}}
\def\bp{{\textbf p}}

\def\beq{\begin{eqnarray}}
\def\eeq{\end{eqnarray}}

\SetKwInput{KwMode}{Mode}
\SetKwInput{KwParameters}{Parameters}

\def\by{\textbf{y}}

\usepackage{sectsty}
\sectionfont{\normalsize}
\subsectionfont{\normalsize\normalfont\itshape}
\newtheoremstyle{thm}
{12pt}
{12pt}
{\itshape}
{}
{\itshape\bfseries}
{}
{1em}
{}

\theoremstyle{thm}
\newtheorem{theorem}{Theorem}

\newtheorem{proposition}[theorem]{Proposition}
\newtheorem{definition}[theorem]{Definition}

\usepackage{xy}
\xyoption{matrix}
\xyoption{frame}
\xyoption{arrow}
\xyoption{arc}

\usepackage{ifpdf}
\ifpdf
\else
\PackageWarningNoLine{Qcircuit}{Qcircuit is loading in Postscript mode.  The Xy-pic options ps and dvips will be loaded.  If you wish to use other Postscript drivers for Xy-pic, you must modify the code in Qcircuit.tex}
\xyoption{ps}
\xyoption{dvips}
\fi

\entrymodifiers={!C\entrybox}

\newcommand{\bra}[1]{{\left\langle{#1}\right\vert}}
\newcommand{\ket}[1]{{\left\vert{#1}\right\rangle}}
\newcommand{\qw}[1][-1]{\ar @{-} [0,#1]}
\newcommand{\qwx}[1][-1]{\ar @{-} [#1,0]}

\newcommand{\gate}[1]{*+<0.5em>{#1} \POS ="i","i"+UR;"i"+UL **\dir{-};"i"+DL **\dir{-};"i"+DR **\dir{-};"i"+UR **\dir{-},"i" \qw}
\newcommand{\meter}{*=<2.0em,1.5em>{\xy ="j","j"-<.778em,.322em>;{"j"+<.778em,-.322em> \ellipse ur,_{}},"j"-<0em,.4em>;p+<.5em,.9em> **\dir{-},"j"+<2.2em,2.2em>*{},"j"-<2.2em,2.2em>*{} \endxy} \POS ="i","i"+UR;"i"+UL **\dir{-};"i"+DL **\dir{-};"i"+DR **\dir{-};"i"+UR **\dir{-},"i" \qw}

\newcommand{\control}{*!<0em,.095em>-=-<.2em>{\bullet}}

\newcommand{\ctrl}[1]{\control \qwx[#1] \qw}

\newcommand{\targ}{*+<.04em,.04em>{\xy ="i","i"-<.50em,0em>;"i"+<.50em,0em> **\dir{-}, "i"-<0em,.47em>;"i"+<0em,.47em> **\dir{-},"i"*\xycircle<.6em>{} \endxy} \qw}
\newcommand{\qswap}{*=<0em>{\times} \qw}

\newcommand{\push}[1]{*{#1}}

\newcommand{\lstick}[1]{*!R!<.5em,0em>=<0em>{#1}}

\newcommand{\Qcircuit}{\xymatrix @*=<0em>}

\begin{document}

\begin{center}
{ \large \bfseries Scalable quantum neural networks by few quantum resources}

\sp

{\sc {Davide Pastorello$^{\diamond,\dagger}$} and  {Enrico Blanzieri}$^{\star,\dagger}$}

\bigskip

{
$^\diamond$ Department of Mathematics, Alma Mater Studiorum - University of Bologna
\\piazza di Porta San Donato 5, 40126 Bologna, Italy  \\ ~~ davide.pastorello3@unibo.it \\[10pt]
$^\star$ Department of Information Engineering and Computer Science, University of Trento, \\via Sommarive 9, 38123 Povo (Trento), Italy\\ ~~ enrico.blanzieri@unitn.it\\[10pt]
$^\dagger$ Trento Institute for Fundamental Physics and Applications, \\via Sommarive 14, 38123 Povo (Trento), Italy\\ ~~}

\end{center}

\vspace{0.5cm}

{\small \section*{Abstract} 
 
This paper focuses on the construction of a general parametric model that can be implemented executing multiple swap tests over few qubits and applying a suitable measurement protocol. The model turns out to be equivalent to a two-layer feedforward neural network which can be realized combining small quantum modules. The advantages and the perspectives of the proposed quantum method are discussed. 

}

\vspace{0.5cm}

\noindent
{\small \textbf{Keywords}: Neural networks; quantum machine learning; scalable quantum computing.}

\vspace{0.5cm}

\noindent
\section{Introduction}

The application of quantum computation to machine learning tasks offers some interesting solutions characterized by a quantum advantage with respect to the classical counterparts in terms of time and space complexity, expressive power, generalization capability, at least on a theoretical level \cite{Wittek, Schuld, Pastorello}. Moreover, quantum machine learning (QML) seems to be a promising path to explore in deep the possibilities offered by quantum machines and to exploit existing and near-term quantum computers for tackling real-world problems.
In this respect, the research activity aimed to the development of novel QML algorithms must face the issue of the lack of large-scaled, fault-tolerant, universal quantum computers. In order to promote QML as an operative technology we should provide solutions that are not too demanding in terms of quantum resources and, in the meanwhile, they supply quantum advantages w.r.t. classical information processing. The present paper focuses on a solution which represents a trade-off between the definition of a quantum model, specifically a quantum neural network, and the availability of few quantum resources as a realistic assumption.


In this work, we assume the availability of a very specific quantum hardware that is able to perform the swap test\footnote{The \emph{swap test} is a simple but useful technique to evaluate the similarity between two unknown pure states \cite{swap}. It is briefly illustrated in the next section.} over two $k$-qubit registers ($k\geq 1$). This assumption is motivated by several proposals of a physical implementation of the swap test then it turns out to be rather realistic for few qubits \cite{swap0, swap1}. Then, we assume to combine many copies of the quantum hardware, as modules, in order to obtain a larger structure for data processing. In particular, we argue that a combined execution of many independent swap tests, performing only local measurements, allows to create a two-layer feedforward neural network presenting some quantum advantages in terms of space complexity with the possibility to input quantum data. The main result is a general construction of a meaningful quantum model based on very few quantum resources whose scalability is allowed by the requirement of quantum coherence over a small number of qubits because a major role is played by measurement processes.    

The swap test is a well-known tool for developing QML techniques, so there are some proposals of quantum neural networks based on it. For instance, there are quantum neurons where the swap test is coupled to the quantum phase estimation algorithm for realizing quantum neural networks with quantum input and quantum output \cite{Zhao, Li}. Our approach is different, we assume only the capability of performing the swap test over few qubits for constructing neural networks where the input is a quantum state and the output is always classical. Our goal is the definition of a quantum model which can be realized assuming the availability of few quantum resources. 

The paper is structured as follows: in section \ref{background} we recall some basic notions about two-layer feedforward neural networks and the definition of swap test as the building block for the proposed modular architecture. In section \ref{sec:model}, we define the parametric model obtained combining small modules acting on few qubits, showing its equivalence to a neural network after a suitable measurement protocol is performed. In section \ref{sec:discussion}, we discuss the impact of the obtained results. In the final section, we draw some concluding remarks highlighting the perspectives of the proposed modular approach.

\section{Background}\label{background}

In this section we introduce some basic notions that are crucial for the definition of the model proposed and discussed in the paper. In particular, we define the structure of the neural network we consider in the main part of the work and the necessary quantum computing fundamentals.

\begin{definition}
Let $X$ and $Y$ be non-empty sets respectively called \textbf{input domain} and \textbf{output domain}. A (deterministic) \textbf{predictive model} is a function
$$y=f(x;\theta)\qquad x\in X,\,\, y \in Y,$$  
where $\theta=(\theta_1,...,\theta_d)$ is a set of real parameters.
\end{definition}

In supervised learning, given a training set $\{(x_1,y_1),\cdots,(x_m,y_m)\}$ of pairs input-output the general task is finding the parameters $\theta$ such that the model assigns the right output to any unlabelled input. The parameters are typically determined optimizing a \emph{loss function} like $\mathcal L(\theta)=\sum_{i=1}^m d(y_i,f(x_i,\theta))$ where $d$ is a metric defined over $Y$. A remarkable example of predictive model, where $X=\bR^n$ and $Y=\bR$, is the following \emph{two-layer neural network} with $h$ hidden nodes:
\beq
f(\bx; \bp, W)=\sum_{i=1}^h p_i \varphi(\bx\cdot \bw_i),
\eeq
where $\bw_i\in\bR^n$ is the $i$th row of the $n\times h$ weight matrix $W$ of the first layer, $p_i$ is the $i$th coefficient of the second layer, and $\varphi:\bR\rightarrow\bR$ is a continuous non-linear function. Despite its simplicity, this model is relevant because it is a \emph{universal approximator} if $\varphi$ is non-polynomial, in the sense that it can approximate, up to finite precision, any Borel function $K\rightarrow \bR$, with $K\subset\bR^n$ compact, provided a sufficient number of hidden nodes is available \cite{hornik, cybenko, pinkus}. 

Let us focus on a quadratic activation function for a two-layer network that is not very used in practice however it has been widely studied \cite{Livni, Du, Soltani,Tan} proving, in particular, that networks with quadratic activation are as expressive as networks with threshold activation and can be learned in polynomial time \cite{Livni}. The network with this activation fails the universality, however the quadratic activation presents the same convexity of the rectified linear unit (ReLU) but with a wider activating zone, so the loss curves of networks using these two activations converge very similarly \cite{Tan}.
In addition to the choice of a quadratic activation, let us consider the \emph{cosine similarity} instead of the dot product as pre-activation:
\beq
\cos(\bx,\bw):=\frac{\bx\cdot\bw}{\parallel \bx\parallel\,\,\parallel\bw\parallel}.
\eeq
The \emph{cosine pre-activation} is a normalization procedure adopted to decrease the variance of the neuron that can be large as the dot product is unbounded. A large variance can make the model sensitive to the change of input distribution, thus resulting in poor generalization. Experiments show that cosine pre-activation achieves better performances, in terms of test error in classification, compared to batch, weight, and layer normalization \cite{cosine}. The aim of the paper is showing that a simple quantum processing, combined in multiple copies, can be used to implement a varying-size two-layer network with cosine pre-activation and quadratic activation.

The main assumption of the work is that few quantum resources are available in the sense that we consider a quantum hardware characterized by a small number of qubits on which we can act with few quantum gates being far from the availability of a large-scaled, fault-tolerant, universal quantum computer. The central object is a well-known tool of QML, the \emph{swap test}, that we assume to have implemented into a specific-purpose quantum hardware, characterized by an arbitrarily small number of qubits, that we call \emph{module}. The questions we address are: \emph{What useful can we do if we are able to build such a module with a small number of qubits? Can we construct a neural network stacking such modules with some interesting properties?    }


As a postulate of quantum mechanics, any quantum system is described in a Hilbert space. We consider only the finite-dimensional case where a Hilbert space $\sH$ is nothing but a complex vector space equipped with an \emph{inner product} $\langle\,\,\,|\,\,\,\rangle:\sH\times\sH\rightarrow\bC$. A quantum system is described in the Hilbert space $\sH$ in the sense that its physical states are in bijective correspondence with the projective rays in $\sH$. Then a (pure) quantum state can be identified to a normalized vector $\ket\psi\in\sH$ up to a multiplicative phase factor. In other words, quantum states are described by one-dimensional orthogonal projectors of the form $\rho_\psi=\ket\psi\bra\psi$ where $\bra\psi$ is an element of the dual of $\sH$.

\begin{definition}\label{def:measurement}
Let $\sH$ be a (fnite-dimensional) Hilbert space and $I$ a finite set. A {\bf measurement process} is a collection of linear operators $\mathcal M=\{E_\alpha\}_{\alpha\in I}$ such that $E_\alpha\geq 0$ for each $\alpha\in I$ and $\sum_\alpha E_\alpha =\bI$, where $\bI$ is the identity operator on $\sH$.

Given an orthonormal basis $\{\ket i\}_{i=1,...,n}$ of $\sH$, the {\bf measurement in the basis} $\{\ket i\}_{i=1,...,n}$ is the measurement process defined by $\mathcal M=\{\ket i\bra i\}_{i=1,...,n}$.

\end{definition}

\noindent
The probability of obtaining the outcome $\alpha\in I$ by the measurement $\mathcal M=\{E_\alpha\}_{\alpha\in I}$ on the system in the state $\ket\psi$ is $\mathbb P_\psi(\alpha)=\langle\psi| E_\alpha\psi\rangle$. If $\mathcal M$ is a measurement in the basis $\{\ket i\}_{i=1,...,n}$ then the probability reduces to $\mathbb P_\psi (i)=|\langle\psi|i\rangle|^2 $.
A \emph{qubit} is any quantum system described in the Hilbert space $\bC^2$, denoting the standard basis of $\bC^2$ by $\{\ket 0, \ket 1\}$ we have the \emph{computational measurement} on a single qubit $\mathcal M=\{\ket0\bra0,\ket1\bra1\}$. 

Another postulate of quantum mechanics reads that, given two quantum systems $S_1$ and $S_2$ described in the Hilbert spaces $\sH_1$ and $\sH_2$ respectively, the Hilbert space of the composite system $S_1+S_2$ is given by the tensor product Hilbert space $\sH_1\otimes\sH_2$. Therefore there are states of the product form $\ket{\psi_1}\otimes \ket{\psi_2}\in\sH_1\otimes\sH_2$ (in contract notation: $\ket{\psi_1\psi_2}$) called \emph{separable states} and states that cannot be factorized called \emph{entangled states}. The latter describe a kind of non-local correlations without a counterpart in classical physics which are the key point of the most relevant advantages in quantum information processing.


Encoding classical data into quantum states is a crucial issue in quantum computing and specifically quantum machine learning. One of the most popular quantum encoding is the so-called \emph{amplitude encoding}. Assume we have a data instance represented by a complex vector $\bx\in\bC^n$ then it can be encoded into the amplitudes of a quantum state:
\beq\label{amplitude}
\ket\bx=\frac{1}{\parallel \bx\parallel}\sum_{i=1}^n x_i \ket i,
\eeq  
where $\{\ket i\}_{i=1,...,n}$ is a fixed orthonormal basis, called \emph{computational basis}, of the Hilbert space of the considered quantum physical system. If the quantum system is made up by $k$ qubits, say a \emph{$k$-qubit register}, with Hilbert space $\sH=(\bC^2)^{\otimes k}$ then for storing $\bx\in\bC^n$ we need $k=\lceil\log n\rceil$ qubits and the commonly used computational basis is $\{\ket d: d\in\{0,1\}^k\}$. Therefore, amplitude encoding exploits the exponential storing capacity of a quantum memory, but it does not allow the direct retrieval of the stored data. Indeed, the amplitudes cannot be observed, and only the probabilities $\lvert x_i\rvert^2/\parallel\! \bx\!\parallel^2$ can be estimated measuring the state $\ket\bx$ in the computational basis. An immediate consequence of the definition is that the inner product of two encoding states corresponds to the cosine similarity of the encoded vectors:
\beq
\langle\bx|\by\rangle=\cos(\bx,\by)\qquad \forall \bx,\by\in\bC^n.
\eeq

Within the quantum circuit model, quantum states are processed by quantum gates. Let us consider the \emph{Hadamard gate}, that is a unitary operator on $\bC^2$, defined in terms of the computational basis by:
\beq
H \ket a =\frac{\ket 0+(-1)^a\ket 1}{\sqrt 2}\qquad a\in\{0,1\}.
\eeq

\noindent
The \emph{controlled NOT} (CNOT) gate is a 2-qubit gate, that is a unitary operator on $\bC^2\otimes\bC^2$, defined by:
\beq
U_{CNOT}\, \ket{ab}:=\ket {a\,(a\oplus b)} \qquad a,b\in\{0,1\},
\eeq
where $\oplus$ is the sum modulo 2. Its circuital symbol is:
$$
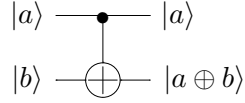
\Qcircuit @C=1em @R=1.5em {
\ket a& &\ctrl{1} &\qw& \ket a \\
\ket b& &\targ &\qw & \hspace{.7cm}\ket {a\oplus b}
}
$$
The \emph{swap gate} is a 2-qubit gate which exchanges the input qubits, its circuital definition in terms of CNOT gates is:
$$
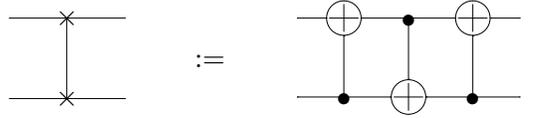
\Qcircuit @C=2em @R=2.8em {
& \qswap & \qw& \\
& \qswap\qwx & \qw &
} 
\Qcircuit @C=2em @R=1.5em {
&  & \\
& \lstick{:=}&\\
&  &
}
\Qcircuit @C=1em @R=1.5em {
& \targ&\ctrl{1} &\targ&\qw \\
& \ctrl{-1}&\targ &\ctrl{-1}&\qw &.
}
$$
The controlled swap gate is called \emph{Fredkin gate}, its action can be expressed in the following form in the computational basis of three qubits, where the first qubit is the control:
\beq
\mathsf F\ket{abc}= (1-a)\ket{abc}+a\ket{acb}\qquad a,b,c\in\{0,1\}.
\eeq

\noindent
The swap test is a very simple procedure to estimate the probability $|\langle\psi|\varphi\rangle|^2$, given two unknown pure states $\ket\psi$ and $\ket\varphi$ \cite{swap}. The procedure requires two applications of the Hadamard gate and a controlled swap operation between two multiple qubit registers generalizing the Fredkin gate. It is straightforward to show that the controlled swap of two $k$-qubit registers can be implemented with $k$ Fredkin gates controlled by the same qubit.
Consider two $k$-qubit registers initialized in $\ket\psi$ and $\ket\varphi$ and a control qubit prepared in $\ket 0$. Assume to act with the following circuit on the $2k+1$ qubits:
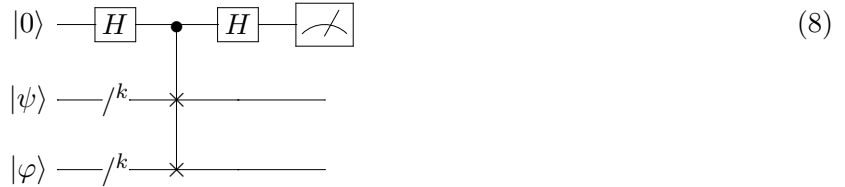
\beq\label{swaptest}
\Qcircuit @C=1.3em @R=1.3em {
\ket 0\qquad &\gate{H}& \ctrl{1} & \gate{H}& \meter\\
\ket\psi\qquad&\push{/^k }\qw& \qswap & \qw& \qw\\
\ket\varphi\qquad &\push{/^k }\qw& \qswap\qwx & \qw &\qw
} 
\eeq

\sp

\noindent
after a simple and well-known calculation we have that the probability of obtaining the outcome $0$, measuring the control qubit in the computational basis, is: 
\beq\label{P0}
\mathbb P(0)=\frac{1}{2}(1-|\langle\psi|\varphi\rangle|^2).
\eeq
Since the probability is estimated as a relative frequency in a Bernoulli trial, the number of runs of the circuit (\ref{swaptest}) required to obtain $\mathbb P(0)$ within an error $\epsilon$ is $\mathcal O(\epsilon^{-2})$. 

From the viewpoint of practical realization, several methods have been proposed to implement a linearly optical swap test \cite{Baldazzi, lin.opt.1,lin.opt.2} and techniques for implementing a swap test based on nonlinear optics have been proposed as well \cite{no.lin.opt.1,no.lin.opt.2}. In this work we realistically assume the availability of a swap test over few qubits, simply depicted by the following process diagram that we introduce for convenience of notation below:

\vspace{3cm}

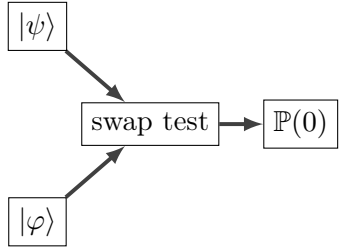
\beq\label{ST}
\quad
\eeq

\vspace{-2.5cm}

\begin{center}

\begin{tikzpicture}
\node[draw](v0) at (-1.5,1.3) {$\ket\psi$};
\node[draw] (v2) at (-1.5,-1.3) {$\ket\varphi$};
\node[draw] (v3) at (0,0) {swap test};
\node[draw] (v1) at (2,0) {$\mathbb P(0)$};
  \Edge[Direct](v2)(v3)(v2)
 \Edge[Direct](v3)(v1);
 \Edge[Direct](v0)(v3)

\end{tikzpicture}

\end{center}

\noindent
where $\ket\psi$ and $\ket\varphi$ are $k$-qubit states and $\mathbb P(0)$ is given by  (\ref{P0}). In the following, we argue that independent swap tests between  $k$-qubit registers and a suitable measurement protocol over the control qubits can be used  for implementing a two-layer feedforward neural network. 

An alternative way to obtain the same estimate of the square of the inner product is worth to be considered. 
The swap test, given the quantum states $\ket\bx,\ket\bw\in(\bC^2)^{\otimes k}$ that encode the input and weight vectors $\bx,\bw$ into the amplitudes w.r.t. the computational basis, returns the probability $\frac{1}{2}(1-|\langle\bx|\bw\rangle|^2)=\frac{1}{2}(1-\cos^2(\bx,\bw))$, by (\ref{P0}), that is the output of a perceptron with a quadratic activation and cosine pre-activation. Alternatively, an analogous parametric model can be defined also considering other encodings, in general given by the action of parametric circuits $U_\bx$ and $V_\bw$ so that $\ket\bx=U_\bx\ket 0$ and $\ket\bw=V_{\bw}\ket 0$. For the amplitude encoding, the number of elementary gates required for the implementations of $U_\bx$ and $V_\bx$ is $\mathcal O(2^k)$. The output $|\langle\bx|\bw\rangle|^2$ corresponds to the probability of obtaining the bit string $0\in\{0,1\}^n$ by the computational measurement over the $k$-qubit state $\ket\Psi_{\bx,\bw}=U_\bw U_\bx\ket 0$. In fact: $\mathbb P_{\Psi_{\bx,\bw}}(0)=|\langle 0|\Psi_{\bx,\bw}\rangle|^2=|\langle 0|U^\dagger_\bw U_\bx\ket 0|^2 =|\langle\bw|\bx\rangle|^2$. Therefore, the parametric model can be simply implemented by the circuit $U=U_\bw U_\bx$ over $k$ qubits initialized in $\ket 0$
 instead of the $2k+1$ qubits required by the swap test. However, the overall number of the measurements required by the circuit $U=U_\bw U_\bx$ is $\mathcal O(k)$. The requirement of measuring only a single qubit of the swap test is particularly convenient in the case of multiple perceptrons (each implemented via a swap test) where their outputs are combined probabilistically as discussed in the next section. Since each module requires only a single-qubit measurement, the overall computation is simplified. This enables the construction of a weighted sum of the outputs without needing simultaneous measurements on all qubits in the system, making the process more efficient and scalable, especially in photonic architectures where both amplitude encoding and swap test can be efficiently implemented \cite{Baldazzi,lin.opt.1}.

\section{The proposed model}\label{sec:model}

Let us consider $m$ copies of the process diagram (\ref{ST}), that we call \emph{modules}, each module presents two $k$-qubit registers and one auxiliary qubit to perform the swap test. Then the total number of qubits, including the controls, is $m(2k+1)$. Let $\bx\in\bR^N$ be an input vector and $\bw\in\bR^N$ be a vector of parameters, if $k\geq \log N$ then a single module is sufficient to store $\bx$ and $\bw$ in its registers and the action of the swap test, which returns the probability $\frac{1}{2}(1-\cos^2(\bx,\bw))$ as the value of a non-linear function over $\cos(\bx,\bw)$, realizes an elementary perceptron with quadratic activation and cosine  pre-activation. However, we are not interested in implementing a single perceptron and according to the general assumption of few quantum resources available, we consider the case $k<\log N$, so in each module we can encode the input and weight vectors only partially.

\begin{proposition}\label{storing}
Let $\mathcal R$ be a $k$-qubit register. Let $N>2^k$, the number $m$ of non-interacting\footnote{In this context, \emph{non-interacting registers} means that any operation on the registers, including those required for state preparation, must be performed locally on each register. Therefore, no entanglement can be produced among different $k$-qubit registers.} copies of $\mathcal R$ needed for storing $\bx\in\bC^N$ into a state $\bigotimes_{i=1}^m \ket{\psi_i}$ by amplitude encoding is:
\beq
m=\mathcal O(N2^{-k}).
\eeq  

\end{proposition}

\begin{proof}
In each register, we can encode at most $2^k$ entries of $\bx$ into the amplitudes of a $k$-qubit state. Therefore the minimum number $m$ of registers for storing $\bx$ into a pure state, where the registers are non-entangled with each other, satisfies $(m-1)2^k<N<m2^k$.

\end{proof}

Considering $m$ independent modules, we can encode at most $m2^k$ input entries and $m2^k$ weights into the available qubits (divided into the \emph{input register} and the \emph{weight register} of each module). Therefore, by proposition \ref{storing}, the total number of qubits required to encode and process $\bx\in\bR^N$ is: 
\beq
n = \mathcal O\left(N(2k+1)2^{-k}\right).
\eeq 
Since a single module does not have an input register large enough to encode all the entries of the input $\bx$ (i.e. $k<\log N$), we have to settle for a piecewise amplitude encoding parallelizing the modules, as a consequence the total number of qubits scales linearly in $N$, however there is an exponential decay in the size of a single module which can be relevant in specific cases. For instance, if $N=128$ and we have a single module with a $5$-qubit input register and a $5$-qubit weight register then 4 modules and 44 qubits are required to store an input vector and a weight vector. Note that in this context 44 qubits is not considered a \emph{large number} because the coherence must be maintained only inside a single module (that is 11 qubits, two 5-qubit registers and 1 control qubit). Thus, the parallelization of modules with few qubits cannot reach of course an exponential storing capacity but there are cases with realistic values of $k$ and $N$ that can be addressed efficiently by modularity.

In order to introduce the parametric model that we can implement executing swap tests in parallel, let us partition the input vector $\bx\in\bR^N$ and the weight vector $\bw\in\bR^N$ into pieces to feed the $m$ modules (each one with $2k+1$ qubits) assuming that $N=m2^k$ without loss of generality: 
\beq
\bx=(\bx^{(1)},...,\bx^{(m)}) \quad\mbox{where}\quad \bx^{(l)}=(x_{(l-1)2^k+1},...,x_{l2^k}),\qquad l=1,...,m
\eeq
$$
\bw=(\bw^{(1)},...,\bw^{(m)}) \quad\mbox{where}\quad \bw^{(l)}=(w_{(l-1)2^k+1},...,w_{l2^k}),\qquad l=1,...,m
$$

\noindent
$\bx^{(l)}$ can be stored into the $k$-qubit {input register} of the $l$th module and $\bw^{(l)}$ can be stored into the $k$-qubit {weight register} of the $l$th module, then the computation of the norms $\parallel\!\bx^{(l)}\!\parallel$ and $\parallel\!\bw^{(l)}\!\parallel$, for any $l=1,...,m$, are required for the encoding.
The parallelization of the modules to perform $m$ swap tests can be easily visualized as follows:

\vspace{3cm}

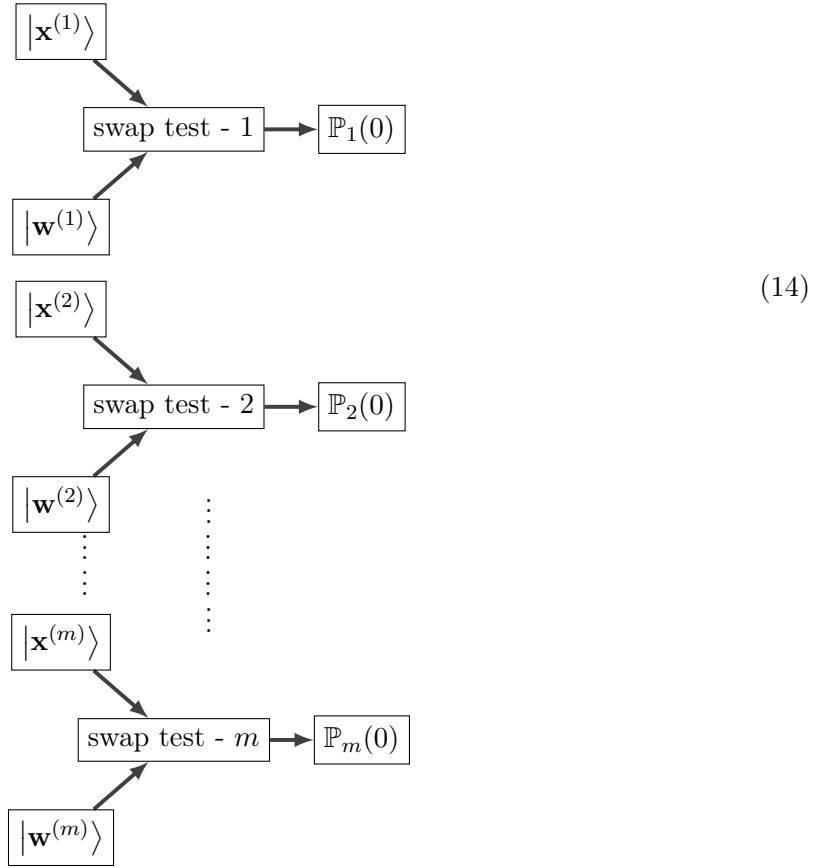
\beq\label{diagram}
\quad
\eeq

\vspace{-4.5cm}

\begin{center}

\begin{tikzpicture}
\node[draw](v0) at (-1.5,1.3) {$\ket{\bx^{(1)}}$};
\node[draw] (v2) at (-1.5,-1.3) {$\ket{\bw^{(1)}}$};
\node[draw] (v3) at (0,0) {swap test - 1};
\node[draw] (v1) at (2.5,0) {$\mathbb P_1(0)$};
  \Edge[Direct](v2)(v3)
 \Edge[Direct](v3)(v1);
 \Edge[Direct](v0)(v3)

\end{tikzpicture}

\vspace{0.3cm}

\begin{tikzpicture}
\node[draw](v0) at (-1.5,1.3) {$\ket{\bx^{(2)}}$};
\node[draw] (v2) at (-1.5,-1.3) {$\ket{\bw^{(2)}}$};
\node[draw] (v3) at (0,0) {swap test - 2  };
\node[draw] (v1) at (2.5,0) {$\mathbb P_2(0)$};
  \Edge[Direct](v2)(v3)
 \Edge[Direct](v3)(v1);
 \Edge[Direct](v0)(v3)

\end{tikzpicture}

\vspace{-1.4cm}

$$\vdots$$

\vspace{-1cm}

$$\vdots\qquad\qquad \vdots\qquad\qquad\,\,$$

\vspace{-1cm}

$$\vdots\qquad\qquad \vdots\qquad\qquad\,\,$$

\vspace{-1cm}

$$\vdots$$

\vspace{-0.5cm}

\begin{tikzpicture}
\node[draw](v0) at (-1.5,1.3) {$\ket{\bx^{(m)}}$};
\node[draw] (v2) at (-1.5,-1.3) {$\ket{\bw^{(m)}}$};
\node[draw] (v3) at (0,0) {swap test - $m$};
\node[draw] (v1) at (2.5,0) {$\mathbb P_m(0)$};
  \Edge[Direct](v2)(v3)
 \Edge[Direct](v3)(v1);
 \Edge[Direct](v0)(v3)

\end{tikzpicture}

\end{center}

\vspace{0.5cm}

\noindent
where $\mathbb P_l(0)$ is the probability of measuring 0 over the $l$th control qubit which can be estimated within an error $\epsilon$ with $\mathcal O(\epsilon^{-2})$ runs. The multiple swap tests depicted in (\ref{diagram}) can be seen as a layer of perceptrons with connectivity limited by the mutual independence of the modules. On this first elementary structure, we construct the second layer applying a specific  \emph{local measurement protocol} to the $m$ control qubits.

{By definition \ref{def:measurement}, a measurement process provides a classical probability distribution over the set of possible outcomes. In operational terms, a measurement process is repeated $\mathcal N$ times and a string of outcomes $x\in\bR^\mathcal N$ is stored, where the probability distribution is estimated in terms of relative frequencies. In the following definition, let us assume that only $\widetilde{\mathcal N}<\mathcal N$ outcomes are stored after $\mathcal N$ repetitions of a considered measurement process $\mathcal M$. 
\begin{definition}
 Let $\mathcal M$ be a measurement process. Assume to repeat the measurement $\mathcal N$ times storing $\widetilde{\mathcal N}<\mathcal N$ outcomes, then $\mathcal M$ is called {\bf lossy measurement} with {\bf efficiency} $p=\widetilde{\mathcal N}/\mathcal N$.
 \end{definition}
\noindent
The measurement protocol that we need for constructing the proposed model is characterized by local measurements over single qubits characterized by given efficiencies. 
\begin{definition}
Given a system composed by $m$ qubits, a {\bf local measurement protocol} is a collection $\{(\mathcal M_l,p_l)\}_{l=1,...,m}$ where $\mathcal M_l$ is a lossy measurement process on the $l$th qubit and $p_l$ its efficiency.
\\
If $\mathcal M_l=\{\ket0\bra0,\ket1\bra1\}$ for any $l=1,...,m$ then $\{(\mathcal M_l,p_l)\}_{l=1,...,m}$ is called {\bf computational measurement protocol}.
\end{definition}
In the case of the computational measurement protocol over $m$ qubits, the outcomes provided by a single shot of the protocol form a string $b=\{\emptyset,0,1\}^m$, where the symbol $\emptyset$ denotes a missing outcome. In the case the measurement processes $\mathcal M_l$ are lossless, the efficiency $p_l$ can be controlled selecting the number $\widetilde {\mathcal N}$ of times a measurement $\mathcal M_l$  is actually performed  when the overall protocol is repeated $\mathcal N$ times. From the physical viewpoint, one can also select the efficiency of the measurements controlling the losses of employed channels.}

Given the $m$ modules executing the multiple (independent) swap tests depicted in diagram (\ref{diagram}), let us define a model with parameters given by the weights $\bw\in\bR^N$ and the efficiencies $\textbf p=(p_1,...,p_m)$ such that the prediction associated to the input $\bx\in\bR^N$ is given by:
\beq\label{model}
f(\bx;\bw,\bp)=\frac{\mN_0}{\mN},
\eeq
where $\mN$ is the number of repetitions of the computational measurement protocol, with efficiencies $\bp=(p_1,...,p_m)$, on the control qubits and $\mathcal N_0$ is the total number of obtained zeros. 

\vspace{0.5cm}

\begin{proposition}\label{prop}
The model (\ref{model}) is equivalent, up to an error $\epsilon=\mathcal O(1/\sqrt \mN)$, to the following two-layer feedforward neural network with cosine pre-activation and activation $\varphi(z)=\frac{1}{2}(1-z^2)$ in the hidden neurons:

\vspace{3.0cm}

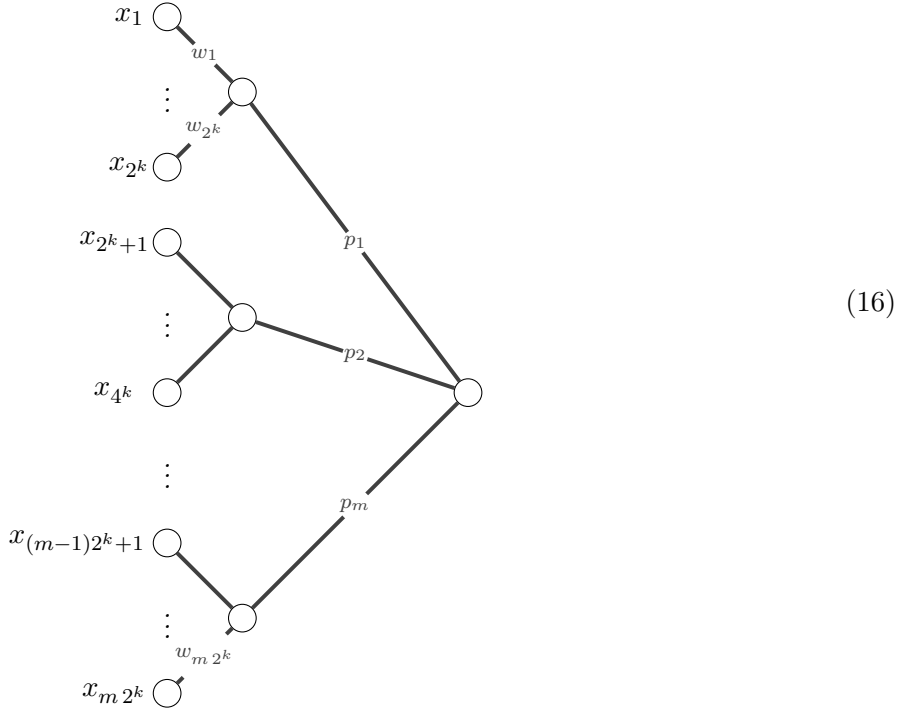
\beq\label{NN}
\qquad
\eeq

\vspace{-5cm}

\begin{center}

\begin{tikzpicture}
\node (w0) at (-1.5,1) {$x_1$};
\node[draw, circle](v0) at (-1,1) {};
\node (vv) at (-1,0) {$\vdots$};
\node (w2) at (-1.5,-1) {$x_{2^k}$};
\node[draw, circle] (v2) at (-1,-1){};
\node[draw, circle] (v3) at (0,0) {$$};

\node (w4) at (-1.7,-2) {$x_{2^k+1}$};
\node[draw, circle](v4) at (-1,-2) {};
\node (vvv) at (-1,-3) {$\vdots$};
\node (w5) at (-1.7,-4) {$x_{2^{(k+1)}}$};
\node[draw, circle] (v5) at (-1,-4){};
\node[draw, circle] (v6) at (0,-3) {$$};
  \Edge(v6)(v5)
 \Edge(v4)(v6)
  \Edge[label=$w_{2^k}$](v3)(v2)
 \Edge[label=$w_1$](v0)(v3) 

\node(V) at (-1,-5) {$\vdots$};

\node (w7) at (-2.2,-6) {$x_{(m-1)2^k+1}$};
\node[draw, circle](v7) at (-1,-6) {};
\node (vvvv) at (-1,-7) {$\vdots$};
\node (w8) at (-1.7,-8) {$x_{m\,2^k}$};
\node[draw, circle] (v8) at (-1,-8){};
\node[draw, circle] (v9) at (0,-7) {$$};
  \Edge[label=$w_{m\,2^k}$](v9)(v8)
 \Edge(v7)(v9)

\node(P) at (4.5,-3.9) {}; 

\node[draw, circle](VV) at (3,-4) {};
\Edge[label=$p_1$](VV)(v3)
\Edge[label=$p_2$](VV)(v6)
\Edge[label=$p_m$](VV)(v9)

\end{tikzpicture}

\end{center}

\end{proposition}

\begin{proof}

By the application of the computational measurement protocol with efficiencies $\{p_l\}_{l=1,...,m}$ to the scheme (\ref{diagram}), the overall probability of obtaining the outcome $0$ on the $l$th control qubit is:
\beq
\widehat{\mathbb P}_l(0)=p_l\,\mathbb P_l(0)=\frac{p_l}{2}\left(1-\cos^2(\bx^{(l)},\bw^{(l)})\right).
\eeq
The probability can be estimated, up to an error $\epsilon$, by the relative frequency $\mathcal N_0^{(l)}/\mN$ where $\mathcal N_0^{(l)}$ is the number of zeros obtained on the $l$th qubit by $\mN=\mathcal O(\epsilon^{-2})$ repetitions of the computational measurement protocol. By construction, the output of the neural network (\ref{NN}) corresponds to the value $\sum_l \widehat{\mathbb P}_l(0)$ that can be estimated by:
$$\sum_{l=1}^m\frac{\mathcal N_0^{(l)}}{\mN}=\frac{\mathcal N_0}{\mN}.$$
Therefore, for any input $\bx$ and any choice of the parameters $\bw$ and $\bp$ the output $\mathcal N_0/\mN$ of the model (\ref{model}) corresponds to the output of the neural network (\ref{NN}) up to an error $\epsilon=\mathcal O(1/\sqrt \mN)$.

\end{proof}

Proposition \ref{prop} states that we can construct feedforward neural networks by combination of swap tests as modules with fixed size, i.e. defined over qubit registers with given number of qubits. Let us note that the number of hidden and output neurons can be arbitrarily increased. In fact, with additional runs of the swap tests, by permutation of the input entries and the re-initialization of the weights $\bw$ one can realize the connections between the input neurons and an arbitrary number of hidden neurons under the topological constraint given by the size of the modules. Collecting the outcomes obtained by the computational measurement protocol with varying efficiencies, we can define the connections from the hidden neurons and an arbitrary number of output nodes. In this respect, the following proposition enables the implementation of a complete feedforward neural network by a suitable measurement protocol.

\begin{proposition}\label{prop1}

Given a single module of size $k$ and the positive integers $R$ and $Q$, let be $\bw_1,...,\bw_R\in\bR^{2^k}$, and $p_{rq}\in[0,1]$ for any $r=1,...,R$ and $q=1,...,Q$. Let us consider the following algorithm:

\begin{algorithm}[H]
\footnotesize

\vspace{0.5cm}

\KwIn{Input vector $\bx\in\bR^{2^k}$} 
\KwResult{Output vector $\by\in\bR^{Q}$} 
\ForEach{$r\gets 1,...,R$}{
\ForEach{$q\gets 1,...,Q$}{run the swap test $\mN$ times on $\ket\bx$ and $\ket{\bw_{r}}$ and measure the control qubit with efficiency $p_{rq}$\; collect the outcomes $b_{rq}\in\{\emptyset,0,1\}^{\mN}$\;
}}
\ForEach{$q\gets 1,...,Q$}
{construct the concatenation $b_q=b_{1q}b_{2q}\cdots b_{Rq}\in\{\emptyset,0,1\}^{R\times \mN}$\;
{$y_q\gets\frac{\mathcal N_{0,q}}{R\mN}$ \qquad[$\mathcal N_{0,q}$ is the number of zeros in the string $b_q$]}
}
\Return $\by=(y_1,...,y_Q)$

\vspace{0.5cm}

\label{QK-medians}
\end{algorithm}

\noindent
The output $\by\in\bR^Q$ corresponds, up to an error $\mathcal O(1/\sqrt {R\mN})$, to the output of a fully connected feedforward neural network, with cosine pre-activation and activation $\varphi(z)=\frac{1}{2}(1-z^2)$, with $2^k$ input nodes, $R$ hidden nodes, and $Q$ output nodes.

\end{proposition}

\begin{proof}
Let us prove the claim simply for $k=R=Q=2$, since the generalization is straightforward.  Consider a module with $k=2$ and the execution of the following steps ($Q=2$): \\
\\
1. run the swap test $\mN$ times with weights $\bw_1\in\bR^{4}$ and measure the control qubit with efficiency $p_{11}$. Collect the outcomes $b_{11}\in\{\emptyset,0,1\}^{\mN}$;\\
2. run the swap test $\mN$ times, with the same weights and measure the control qubit with efficiency $p_{12}$. Collect the outcomes $b_{12}\in\{\emptyset,0,1\}^{\mN}$ ;
\\
\\
By proposition \ref{prop}, computing the relative frequency of 0 over $b_{11}$ and $b_{12}$ respectively we obtain the estimations of the probabilities $\frac{p_{11}}{2}(1-\cos^2(\bw_1, \bx))$ and $\frac{p_{12}}{2}(1-\cos^2(\bw_2,\bx))$ implementing the following network up to an error $\mathcal O(1/\sqrt \mN)$:

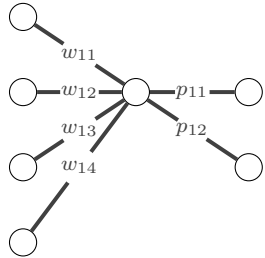
\begin{center}

\begin{tikzpicture}
\node[draw, circle](i1) at (-1.5,1) {};
\node[draw,circle] (i2) at (-1.5,0) {};
\node[draw, circle] (i3) at (-1.5,-1){};
\node[draw, circle] (i4) at (-1.5,-2) {$$};
\node[draw,circle] (h1) at (0,0){};
\node[draw,circle] (o1) at (1.5,0){};
\node[draw,circle] (o2) at (1.5,-1){};
\Edge[label=$w_{11}$](i1)(h1)
\Edge[label=$w_{12}$](i2)(h1)
\Edge[label=$w_{13}$](i3)(h1)
\Edge[label=$w_{14}$](i4)(h1)
\Edge[label=$p_{11}$](h1)(o1)
\Edge[label=$p_{12}$](h1)(o2)

\end{tikzpicture}
\end{center}

\noindent
Repeating the steps 1 and 2 with weights $\bw_2\in\bR^4$ and efficiencies $p_{21}$ and $p_{22}$, we can collect also the outcomes $b_{21}, b_{22}\in\{\emptyset,0,1\}^{\mN}$. Computing the relative frequency of 0 in the strings $b_{21}$ and  $b_{22}$, we can implement another network as well: 

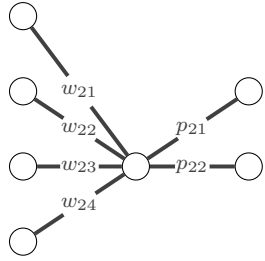
\begin{center}

\begin{tikzpicture}
\node[draw, circle](i1) at (-1.5,1) {};
\node[draw,circle] (i2) at (-1.5,0) {};
\node[draw, circle] (i3) at (-1.5,-1){};
\node[draw, circle] (i4) at (-1.5,-2) {$$};
\node[draw,circle] (h2) at (0,-1){};
\node[draw,circle] (o1) at (1.5,0){};
\node[draw,circle] (o2) at (1.5,-1){};
\Edge[label=$w_{21}$](i1)(h2)
\Edge[label=$w_{22}$](i2)(h2)
\Edge[label=$w_{23}$](i3)(h2)
\Edge[label=$w_{24}$](i4)(h2)
\Edge[label=$p_{21}$](h2)(o1)
\Edge[label=$p_{22}$](h2)(o2)

\end{tikzpicture}
\end{center}

\noindent
In order to obtain a complete network, we can combine the two networks above ($R=2$), providing the outputs:
$$y_1= \frac{p_{11}}{2}(1-\cos^2(\bw_1,\bx))+\frac{p_{21}}{2}(1-\cos^2(\bw_2,\bx)) ,$$
$$y_2= \frac{p_{12}}{2}(1-\cos^2(\bw_1, \bx))+\frac{p_{22}}{2}(1-\cos^2(\bw_2, \bx)) .$$
\\
For estimating the sum $y_1$ of the two probabilities we take the sum of the relative frequencies of the outcome 0 in the string $b_{11}$ and in the string $b_{21}$ which obviously corresponds to the relative frequency of the outcome zero in the concatenation $b_{11}b_{21}$. In the same way, we estimate $y_2$ as the relative frequency of 0 in the string $b_{12}b_{22}$. In other words, given an input $\bx\in\bR^4$ and the parameters $\bw_1, \bw_2, p_{11}, p_{12}, p_{21}, p_{22}$, we have a model that outputs:
$$ y_1=\frac{\mathcal N_{0,1}}{2\mN}\qquad\mbox{and}\qquad y_2=\frac{\mathcal N_{0,2}}{2\mN},$$
where $\mathcal N_{0,i}$ is the number of zeros in the concatenation $b_i=b_{1i}b_{2i}$. Therefore, $y_1$ and $y_2$ correspond, up to an error $\mathcal O(1/\sqrt {2\mN})$, to the outputs of the fully connected network:
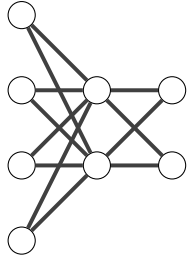
\begin{center}

\begin{tikzpicture}
\node[draw, circle](i1) at (-1,1) {};
\node[draw,circle] (i2) at (-1,0) {};
\node[draw, circle] (i3) at (-1,-1){};
\node[draw, circle] (i4) at (-1,-2) {$$};
\node[draw,circle] (h1) at (0,0){};
\node[draw,circle] (h2) at (0,-1){};
\node[draw,circle] (o1) at (1,0){};
\node[draw,circle] (o2) at (1,-1){};
\Edge (i1)(h1)
\Edge (i1)(h2)
\Edge (i2)(h1)
\Edge (i2)(h2)
\Edge (i3)(h1)
\Edge (i3)(h2)
\Edge (i4)(h1)
\Edge (i4)(h2)
\Edge (h1)(o1)
\Edge (h1)(o2)
\Edge (h2)(o1)
\Edge (h2)(o2)

\end{tikzpicture}
\end{center}

\noindent
where the 8 connections in the first layer are weighted by the components of $\bw_1$ and $\bw_2$ and the 4 connections in the second layer are weighted by $p_{11}, p_{12}, p_{21}, p_{22}$.

By direct generalization of the construction above we have that the number of steps $Q$ establishes the number of output neurons and the number of repetitions $R$ corresponds to the number of hidden neurons.

\end{proof}

Let us remark that the complete topology of the network, provided by proposition \ref{prop1}, in the first layer is implied by the fact that we consider a single module for the calculation of the swap test. Otherwise, the connectivity in the first layer depends by the number of the considered modules and their size $k$. On the contrary, the local measurement protocol described in the algorithm  of proposition \ref{prop1} always enables all-to-all connections in the second layer.

\section{Discussion}\label{sec:discussion}

The statements of propositions \ref{prop} and \ref{prop1} provide a general procedure of constructing measurement protocols for the implementation of feedforward neural networks using the outcomes obtained running multiple swap tests between quantum states encoding the input vectors and the weights vector. In particular, proposition \ref{prop} implies that if we are able to perform independent swap test over $k$-qubit registers then a suitable sample is sufficient for obtaining a neural structure for processing data encoded into the amplitudes of input states. Moreover, proposition \ref{prop1} establishes that, by suitable repetitions of the swap test, one can decide the number of hidden and output neurons. { In particular, if the number of features in the considered dataset is sufficiently low, complete neural networks can be implemented as a direct consequence of propositions \ref{prop} and \ref{prop1}.}

{The modular structure of the model naturally enables the possibility of ensemble learning. In fact, an input vector can be trivially partitioned as described in section \ref{sec:model} but also random partitions can be actuated, in this case a random sampling is performed over the training set as done in bootstrap aggregating for instance. Moreover, different $k$-sized modules initialized with the same weights $\bw\in\bR^{2^k}$ can be interpreted as the application of a convolutional layer followed by a local pooling as done in convolutional neural networks.    }

From the viewpoint of time complexity, the performances of the proposed model are bounded by the state preparation for encoding the input $\bx\in\bR^N$ and the weights $\bw\in\bR^N$ into quantum states since the time taken by the measurement protocols is independent from the input size and only related to accuracy and hyperparameters. The general issue of efficient state preparation is a common bottleneck of the current quantum machine learning algorithms, in particular in quantum neural networks. Indeed, the encoding of a data vector $\bx\in\bR^N$ into the amplitudes of a $\log N$-qubit state $\ket\bx$ can be performed in time $\mathcal O(\log N)$ only if a quantum RAM (QRAM) is available otherwise the retrieval of $\ket\bx$ and the initialization of $\ket\bw$ take time $\mathcal O(N)$. 
 From the viewpoint of the space complexity, there is not the expected exponential storage capacity because we assume the availability of a collection of small quantum register where the input can be encoded piecewise. However, the number of required qubits is tamed by an exponential decay in the size $k$ of a single module then, for suitable values of $N$ and $k$, data can be represented efficiently. Moreover, even if the number of qubits is high, we do not require the coherence over all the qubits but only in any module where the swap test is performed.  
In addition to the advantage in the storing capacity (bounded by $k$) the quantum implementation of the proposed scheme offers the possibility of submitting superposition of data to the network. Moreover, even if the modules are \emph{a priori} uncorrelated, an input can be represented by quantum data with entanglement distributed on a high number of qubits, in this case the quantum processing of a widely entangled input creates quantum correlations among the modules. In general, the input of the model can be a quantum state encoding a feature vector as discussed but also non-artificial quantum data entangling the modules giving rise to a non-trivial execution of the neural network.

Let us remark that the global structure of neural network is also preserved when modules are of minimum size ($k=1$), in this case we have the lowest storing capacity and the first layer turns out to be sparse. Nevertheless, the model is still quantum then capable to process quantum data as a two-layer feedforward network. 

\section{Conclusions}

In this paper we have presented a scheme for implementing a two-layer feedforward neural network using repeated swap tests over quantum registers and suitable measurement protocols. The main goal of the work is showing how to create a non-trivial quantum machine learning model when few quantum resources are available.   In particular, we assume the availability of a specific-purpose quantum hardware, that we call \emph{module}, capable of implementing the swap test (this assumption is motivated by several recent proposals of photonic platforms for this task). By combination of an arbitrary amount of these modules, a scalable neural network can be constructed where the first layer is implemented by the quantum processing itself and the second one is obtained in terms of a measurement protocol characterized by varying efficiencies in registering the outcomes. Even if the composition of independent modules does not create entanglement, an entangled state in input introduces quantum correlations in the network.

The empirical evaluation of the proposed model is not addressed in the present paper where only a theoretical construction is provided. From the experimental viewpoint, in the case of amplitude encoding of classical data a simulation of the multiple swap tests is not so useful because it mathematically corresponds to the execution of a two-layer network with quadratic activation that is a well-known object already investigated in deep. Then, experiments on real quantum hardware would be definitely more significant for revealing actual quantum advantages of the discussed proposal. On the contrary, in the case of quantum data where the input can be represented by highly entangled quantum states, a simulation could be interesting but restricted to a small total number of qubits for computational reasons. 
Indeed, in our opinion, the crucial experimental point will be the combination of real photonic hardware, as small modules implementing the swap test over few qubits, in order to quantify the performances of the scalable quantum neural networks constructed by modularity. In this perspective, a realization with photonic platforms can lead to an on-chip embedding of the modules increasing the technological impact of the proposed QML architecture.

\section*{Acknowledgements}

This work was partially supported by project SERICS (PE00000014) under the MUR National Recovery and Resilience Plan funded by the European Union - NextGenerationEU. The authors are grateful to Sebastian Nagies and Emiliano Tolotti whose ongoing work on the implementation provided useful insights.


\end{document}